\documentclass[11pt
]{IEEEtran}
\usepackage{multicol}
\usepackage{amsmath,graphicx, amsfonts, amssymb, amsthm,paralist} 

\usepackage[utf8]{inputenc}
\usepackage[english]{babel}

\usepackage{algorithm}
\usepackage{algorithmic}

\usepackage{mathtools}
\usepackage[usenames,dvipsnames]{color}
\floatstyle{ruled}
\newfloat{algorithm}{tbp}{loa}
\providecommand{\algorithmname}{Algorithm}
\floatname{algorithm}{\protect\algorithmname}

\theoremstyle{definition}
\newtheorem{definition}{Definition}[section]

\newtheorem{theorem}{Theorem}[section]

\newtheorem{lemma}[theorem]{Lemma}

\newcommand{\expect}[1]{\mathbb{E}\left[ #1 \right]}
\usepackage{wrapfig}
\usepackage{bbold}
\usepackage{amsbsy}
\usepackage{extarrows}

\usepackage{accents}
\newlength{\dhatheight}

\newcommand{\gre}{\varepsilon}

\newtheorem{prop}{Proposition}
\newtheorem{claim}[theorem]{Claim}

\theoremstyle{remark}
\newtheorem{question}{Question}
\newtheorem{coro}{Corollary}

\newcommand{\bea}{\begin{array}}
\newcommand{\ena}{\end{array}}
\newcommand{\bds}{\begin {description}}
\newcommand{\eds}{\end {description}}
\newcommand{\bdf}{\begin{definition}}
\newcommand{\blm}{\begin{lemma}}
\newcommand{\edf}{\end{definition}}
\newcommand{\elm}{\end{lemma}}
\newcommand{\bthm}{\begin{theorem}}
\newcommand{\ethm}{\end{theorem}}
\newcommand{\bprp}{\begin{prop}}
\newcommand{\eprp}{\end{prop}}
\newcommand{\bcl}{\begin{claim}}
\newcommand{\ecl}{\end{claim}}
\newcommand{\bcr}{\begin{coro}}
\newcommand{\ecr}{\end{coro}}
\newcommand{\bquest}{\begin{question}}
\newcommand{\equest}{\end{question}}

\newcommand{\rarrow}{{\rightarrow}}

\title{Distributed learning in congested environments with partial information}
\author{Tomer Boyarski and  Amir Leshem and  Vikram Krishnamurthy}
\begin{document}

\maketitle
\begin{abstract}

How can non-communicating agents learn to share congested resources efficiently? This is a  challenging task when the agents can access the same resource simultaneously (in contrast to multi-agent multi-armed bandit problems) and the  resource valuations differ among agents.  We present a fully distributed algorithm for learning to share in congested environments and prove that the agents' regret with respect to the optimal allocation is poly-logarithmic in the time horizon. Performance in the non-asymptotic regime is illustrated in numerical  simulations. The  distributed algorithm has applications 
in 
cloud computing and spectrum sharing.

keywords: Distributed learning, congestion games, poly-logarithmic regret.
\end{abstract}

\section{Introduction}
\label{introduction}
Suppose $N$ agents need to share $M$ resources where $N\gg M$, i.e., in a congested environment. The utility of each agent $n$ at each time instant $t$ depends on the resource $m_n^t$  it chooses and is inversely proportional to the number of other agents that choose the same resource at the same time. Each agent is aware of this time-sharing structure but can only measure its  utility with added sub-Gaussian noise that is i.i.d. in time.

This paper considers a generalization of the multi-agent multi-armed bandit problem for distributed resource sharing. Our formulation uses  a congestion game with incomplete information. The algorithms and analysis in this paper departs significantly from existing works in that we consider simultaneous  resource sharing in heavily loaded systems with multiple non-communicating  agents. 
Specifically, we develop a learning framework for incomplete information congestion games with non-linear utilities. 
To that end, we construct a novel algorithm (Estimate, Negotiate, Exploit) where the learning is divided into epochs of increasing length. Each epoch has a constant exploration phase followed by a negotiation phase which is poly-logarithmic fraction of the following exploitation phase. By proving that the probability of error of the negotiation phase decreases sup-exponentially we obtain the main result which is the  poly-logarithmic regret in Theorem \ref{thm:TotRegBnd}. 

The algorithm extends the work of \cite{marden2014achieving} to incomplete information games and  \cite{bistritz2018distributed} to the heavily congested case where there are fewer resources than agents and there is no collision information but only rewards which depend on the load since multiple agents access each resource simultaneously. 

To elaborate on the motivation for our approach, we note that 
sharing of resources among agents can benefit significantly from cooperation. This is especially important when there is no central management of the resources.  Resource sharing is particularly hard when communication between agents is limited or does not exist. A further complication arises when agents need to learn their individual resource valuations. An example of unknown valuations is the case of multiple servers have different hardware architectures, e.g., GPU-based machines and vector processors, a different number of cores, size of memory, and different communication links to the servers are some examples.  

Resource sharing can be modeled as a non-cooperative  game \cite{owen1968game}, specifically, a congestion game \cite{rosenthal1973class} and even more specifically, a congestion game with agent-specific utilities \cite{milchtaich1996congestion} that arise when agents  are scattered in the physical world. 
These games are used to model a wide variety of applications including routing, load-balancing, and spectrum allocation \cite{pradelski2012learning, suri2004uncoordinated,  cheng2013opportunistic}. With game theory, it is possible to design adaptive agents whose actions optimize the system, even when based on partial and imprecise information.

The welfare of a game is defined as the sum of the rewards. 
Best response algorithms applied to congestion games converge to Pure Nash Equilibria (PNE), that may have low welfare. One way to deal with this problem is to bound the ratio between the best and worst PNE \cite{koutsoupias1999worst}. Another way is to search for an efficient PNE \cite{pradelski2012learning}. Yet another way is to perform the welfare maximization of congestion games in a central unit \cite{blumrosen2007welfare}. 
A fourth way is to simplify the discussion by not considering player-specific rewards \cite{cheng2013opportunistic}. 
Finally, the discussion can be restricted to a special class of congestion games called collision games, where all colliding players receive zero utility \cite{bistritz2018game}.

It has been shown recently that for collision-type  multi-agent multi-armed bandit problems, a distributed assignment problem can be solved without explicit information exchange between agents, as long as the number of resources is larger or equal to the number of agents. Examples include  the  distributed auction algorithm \cite{zafaruddin2019distributed}, swap-based algorithms for stable configurations \cite{avner2019multi, darak2019multi}, or to signal resource valuations \cite{boursier2019practical,bubeck2019non, tibrewal2019multiplayer} and the musical chairs algorithm \cite{rosenski2016multi}.  Similarly, algorithms for the adversarial case have  been studied in   \cite{alatur2020multi}. In these publications, it is assumed that agents who choose the same arm at the same time receive no reward.  In contrast to these works, motivated by applications in spectrum collaboration in ad-hoc wireless networks  \cite{bistritz2018distributed,  zafaruddin2019distributed,bubeck2019non,avner2019multi, tibrewal2019multiplayer, darak2019multi}, cloud computing \cite{tang2018dynamical} and machine scheduling  \cite{tang2018dynamical},  we  deal with the heavily congested regime together with utilities that depend non-linearly on the number of sharing agents.

In contrast to most previous works on multi-agent multi-armed bandit, when more than one of our agents choose the same arm at the same time they suffer a penalty, but still retain some of their reward. 
This occurs since the sharing of arms is done in a round-robin fashion where each agent receives an equal fraction of the time on the shared arm.
 
The arXiv preprint \cite{magesh2019multi} assumes that if more than $N^*$ agents are sharing a resource they all receive zero utility and therefore the total number of users is bounded by $N^*M$ (see Section II in \cite{magesh2019multi}).
This enables a lower bound on the sub-optimality gap which in turn determines the parameter $\varepsilon$.
In contrast, our algorithm does not rely on such assumptions. This fact results in  worse regret since we have no information on the optimality gap. Indeed the proof of our main theorem \ref{thm:TotRegBnd} more complicated.

The work in \cite{bande2019multi} also considers non-zero utilities upon collision, but their scenario is greatly simplified compared to ours in two ways. Firstly, they deal with homogeneous case whereas we deal with the  heterogeneous case (their utilities are not agent-dependent whereas our are agent-dependent). Secondly, their agents have access to collision indicators that greatly aid in finding the solution whereas ours do not have  access to such indicators.

The rest of this paper is organized as follows: Section \ref{sec:Cooperative Repeated Congestion games} sets up the model formulation and describes the problem. Section \ref{sec: Learning to play repeated congestion games}  describes our novel  learning algorithm. Section \ref{sec:regret}
 performs a regret analysis of the learning algorithm.
In Section \ref{sec: simulation results}
we give numerical examples to illustrate the regret of the proposed algorithm compared to the distributed Upper Confidence Bound algorithm and random allocations. 
Some details of the proofs are provided in the appendix.

\section{Distributed Cooperative Sharing of Congested Resources}
\label{sec:Cooperative Repeated Congestion games}
In this section, we define the resource sharing problem. 
We assume that each resource is equally shared among the agents who choose it, e.g., via a round-robin mechanism. This is the simplest mechanism for sharing the resource when the resource is required continuously by all agents. 
Suppose $N$ agents are sharing $M$ resources where $N\gg M$. This makes the resource sharing much more challenging compared to  $N\leq M$; Yet this model is  important in spectrum sharing and cloud computing applications as discussed in Section~\ref{introduction}.

We assume that time is slotted with $t=1,2,\ldots,T$ indexing the time slots (discrete time).  and agents are synchronized to the slots. 
The number of time slots  $T$  is unknown to the agents. 
The single resource chosen by agent $n$ at time $t$ is denoted $m_n^t$. 
The \emph{allocation } at time $t$ is $$ \mathbf{m}^t = (m_1^t,...,m_N^t)$$ The \emph{load} experienced by agent $n$ at time $t$ under allocation  $\mathbf{m}^t$ is the number of agents who chose the same resource (including itself), i.e.,  
\begin{align}
    \ell_n^t = \ell_n(\mathbf{m}^t) \triangleq \sum_{k =1}^N \mathbb{1}(m_k^t = m_n^t).
\end{align}
The utility of agent $n$ with resource $m$ and load 1 is denoted by $U_{n,m,1}$ and is constant in time.
We assume that these utilities are non negative and bounded by $U_{\max}$. 
More generally, the utility of agent $n$ at time $t$ under  allocation $\mathbf{m}^t$ is
\begin{equation}
\label{eq:util_def}
    \begin{split}
        U_n(\mathbf{m}^t) = 
        \frac{U_{n,m_n^t,1}}{\ell_n^t}.
    \end{split}
\end{equation}

The welfare $W$ at time $t$ with allocation $\mathbf{m}^t$ is the sum of the utilities over the $N$ agents:
\begin{equation}
    W^t = W(\mathbf{m}^t) \triangleq \sum _{n=1} ^N U_n(\mathbf{m}^t) = \sum _{n=1} ^N U_n^t . 
\end{equation}
The best and second-best welfares are denoted by 
\begin{equation}
    W^* \triangleq \displaystyle \max _{{\mathbf{m}} } W(\mathbf{m}) \qquad W^{**} \triangleq \displaystyle\max_{\mathbf{m} \neq \mathbf{m}^*} W (\mathbf{m})
\end{equation}
The sub-optimality gap is defined as:
\begin{equation}
    \rho \triangleq \frac{W^*-W^{**}}{2N}
\end{equation}
The optimal allocation  is
\begin{equation}
\label{eq:optimal_allocation_definition}
    \mathbf{m}^* \triangleq \underset{{\mathbf{m}} }{\arg \max} \;  W(\mathbf{m}).
\end{equation}
A crucial property of our model is that of {\em incomplete information}: agents are aware of the time-sharing structure of the utility function that is inversely proportional to their load, but they can not directly observe their utilities, and they do not know their unique utility function parameters $U_{n,m,1} \; \forall 1 \leq m \leq M$. Instead, they observe noisy versions of their utilities known as \emph{sample rewards}. The sample reward of agent $n$ at time $t$ with in allocation $\mathbf{m}^t$ is 
\begin{equation}
\begin{split}
    r_n(\mathbf{m}^t) 
    = U_n(m_n^t) + \nu_n^t,
\end{split}
\end{equation}
where $\nu_n^t$ is zero-mean sub-Gaussian noise, i.i.d. in time and among agents with variance proxy $b$\footnote{Some of the basic properties of sub-Gaussian random variables used in this paper are mentioned in the appendix.}.

Since the utilities of the agents with each resource and load 1 are independently distributed continuous random variables, drawn once at the beginning of the sharing process, the optimal allocation  $\mathbf{m}^*$ is unique with probability 1.

The main performance metric of the above resource sharing processes is the {\em regret}: 
\begin{equation}
    R \triangleq T W^* - \mathbb{E}\left( \sum _{t=1} ^T W^t \right),
\end{equation}
where the expectation $\mathbb{E}$ is taken with respect to the randomness in the rewards as well as the agents' choices.

\section{Learning the optimal allocation}
\label{sec: Learning to play repeated congestion games}
How can the optimal allocation ${\bf m}^*$ defined in (\ref{eq:optimal_allocation_definition}) be learnt by the agents in a distributed way?
This section presents the {\em  Estimation, Negotiation, and Exploitation (ENE)} learning algorithm that achieves a regret that is poly-logarithmic in the number of time steps
\begin{equation*}
    R = O(\log_2^{3+\delta}(T))
\end{equation*}
where $0\leq \delta\leq 1$.
The algorithm divides the $T$ time-slots  of the sharing problem into $J$ epochs of dynamic length. 
As in other related work, e.g.,  \cite{bistritz2018distributed, zafaruddin2019distributed,boursier2019practical,tibrewal2019multiplayer},  each epoch is further divided into phases whose length is also dynamic.

Our proposed  algorithm has three phases: Estimation, Negotiation, and Exploitation. 
\begin{compactenum}
\item
In the first phase each agent individually and distributedly {\em estimates} its utilities with any resource and any load. 
\item In the second phase agents {\em negotiate} over resources without direct communication.
\item In the third phase agents {\em exploit} the allocation they distributedly  decided upon in the previous phase. 
If the first and second phases were successful, the third phase is regret-free. 
\end{compactenum}


Each phase is further divided into blocks. 
The purpose of most blocks is to average out the added sub-Gaussian noise. 
The number of blocks is different for different phases, and their size is also different.
Finally, the blocks are composed of time-steps, which are the underlying time-steps of the sharing process.
For example, the reward of agent $n$ with resource $m$ and load $\ell$ in epoch $j$, phase 2, block $k$, and time-step $\tau$ is denoted $r^{j,2,k,\tau}_{n,m,\ell} = U_{n,m,\ell} + \nu^{j,2,k,\tau}_{n}$ since the sub-Gaussian noise $\nu$ is independent and identically distributed in time and among agents while the utility is a function of $n,m$ and $\ell$ only.

{\bf Estimation phase:}
This phase has $M+1$ blocks. Each block has $j$ time-steps. In block $k$ where $1 \leq k \leq M$ agent $n$ accesses resource $k$ with probability 1. Agent $n$ then estimates its utility with resource $k$ and load $N$ as the average of its rewards from this block from all epochs until now:
\begin{equation}
\label{eq:estimate_u_n_m_N}
    \hat{U}^j_{n,k,N} = \overline{r}_n^{j,1,k} = \frac{\sum_{i=1}^j \sum_{\tau=1}^{i}r_n^{i,1,k,\tau} }{\frac{1}{2} j(j+1)} \overset{j\rightarrow\infty} {\longrightarrow} U_{n,k,N}
\end{equation}
In each time step of block $M+1$ each agent accesses the first resource with probability $1/2$. Agent $n$ denotes the  average of its rewards from this block from all epochs until now by\footnote{See appendix}
\begin{equation}
\label{eq:estimate_u_n_M_plus_1}
\begin{split}
    &\overline{r}_n^{j,1,M+1} =
    \\
    &\frac{
    \sum_{i=1}^j \sum_{\tau=1}^{i} 
    \left(
    \mathbb{1}
    \left(
    m_n^{j,1,M+1,\tau}=1
    \right)
    \cdot r_n^{i,1,M+1,\tau}
    \right)}
    {
    \sum_{i=1}^j \sum_{\tau=1}^{i} \;\,\mathbb{1}
    \left(
    m_n^{j,1,M+1,\tau}=1
    \right)
    \hfill
    } 
    \\
    &\overset{j\rightarrow\infty} {\longrightarrow}
    U_{n,1,N}\cdot 2 \left(1-\frac{1}{2^N}\right)
\end{split}
\end{equation}

Agent $n$ estimates $N$ in epoch $j$ to be
\begin{equation}
\label{eq:estimation_of_the_number_of_agents}
    \hat{N}^j_n=\frac{1}{\ln{(1/2)}}\ln{\left(1 - \frac{\overline{r}_n^{j,1,M+1}}{2\overline{r}_n^{j,1,1}}\right)} 
    \overset{j\rightarrow\infty} {\longrightarrow}
    N
\end{equation}
Agent $n$ then uses its estimate of the number of agents $\hat{N}^j_n$ together with its estimate of its individual utility with resource $m$ and maximal load $\hat{U}^j_{n,m,N}$  to estimate its utility with any load:
\begin{equation}
\label{eq:estimate_u_n_m_ell}
    \hat{U}^j_{n,m,\ell}
    =
    \frac{1}{\ell}\hat{N}\hat{U}^j_{n,m,N}
    \overset{j\rightarrow\infty} {\longrightarrow}
    U_{n,m,\ell}
\end{equation}
The allocation that maximizes the estimated utilities of epoch $j$ is denoted by
\begin{equation}
    \mathbf{m}^{*j} \triangleq \underset{{\mathbf{m}} }{\arg \max} \sum _{n=1} ^N \hat{U}^j_{n,m_n,\ell_n (\mathbf{m})}
    \label{eq:est_opt_alloc}
\end{equation}
The estimated optimal utilities are 
\begin{equation}
\label{eq:opt_est_util}
\begin{split}
    \hat{U}^{*j}_{n} \triangleq \hat{U}^j_{n,m_n^{*j},\ell_n (\mathbf{m}^{*j})}
    \\
    \mathbf{\hat{U}}^{*j} \triangleq \left(\hat{U}^{*j}_{1},...,\hat{U}^{*j}_{N}\right)
\end{split}
\end{equation}

{\bf Negotiation Phase:} 
In the heart of the ENE algorithm is the Negotiation Phase, inspired by \cite{marden2014achieving}. The Negotiation Phase of epoch $j$ is divided into $j^{1+\delta/3}$  load-estimation-blocks\footnote{We believe that using the significantly more complicated techniques from \cite{bistritz2020GOT} we can use only $j^{\delta/3}$ Negotiation Blocks, thereby reducing to total regret of the algorithm to $O(\log^{2+\delta} T)$.} that are each composed of $j^{1+\delta/3}$ time-steps.
Agent $n$ in block $k$ has a mood that is either Content or Discontent and denoted by $S_n^{j,2,k}\in\{C,D\}$.
A Content agent is stable while a Discontent agent is unstable. 
The probability of an individual agent to enter such a Content and Stable state increases with its estimated utility. 
Therefore, the probability of the community to enter an all-Content-all-stable state increases with the Welfare.
On the other hand, the probability to exit an all-Content-all-stable state is constant and independent of the Welfare. 
Hence, during the tail of the Negotiation Phase, the community will spend most of its time in an optimal all-Content-all-stable state with high probability. Agents can then count which resources they visit most frequently during the tail of the Negotiation Phase and use these during the Exploitation Phase. 

In the first block of the Negotiation Phase, all agents are Discontent. 
In block $k$ agent $n$ performs the following actions distributedly and individually without direct communication with its peers:
\begin{enumerate}
    \item 
    Choose a resource. A Discontent agents chooses a resource uniformly at random:
    \begin{equation}
    \mathbb{P}\left(m_n^{j,2,k}=a\right)
    =
    \frac{1}{M},\,\forall\; 1\leq a\leq A.\label{eq:content_action_choice}
    \end{equation}
    
    A Content agent chooses the same resource with high probability and will otherwise explore uniformly:
    \begin{equation}
        \mathbb{P}\left(m_n^{j,2,k}=a\right)
        =
        \Biggl\{\begin{array}{cc}
        1-\varepsilon^{c} &
        a=m_n^{j,2,k-1}\\
        \frac{\varepsilon^{c}}{A-1} &
        a\neq m_n^{j,2,k-1}
        \end{array}.\label{eq:discontent_action_choice}
    \end{equation}
    where $c>N$ is a parameter of algorithm 
    \ref{alg:Estimate, Negotiation, Exploitation}.
    \item Stay with this resource for the rest of this block and collect  $j^{1+\delta/3}$ i.i.d. reward samples.
    \item Averages these reward and denotes the average by $\overline{r}_n^{j,2,k}$.
    \item Estimate load based on the utility estimation from the previous phase:
    \begin{equation}
        \hat{\ell}_n^{j,2,k} \leftarrow{} \underset{ 1 \leq \ell \leq N }{\arg \min} \left| \overline{r}_n^{j,2,k} - \hat{U}^j_{n, m_n^{j,2,k},\ell}  \right|,
    \label{eq:estimate_load}
    \end{equation}
    \item 
    Estimate utility based on the load estimation $\hat{\ell}_n^{j,2,k}$ and the utility estimation of the previous phase:
    \begin{equation}
        \hat{U}_n^{j,2,k}  \leftarrow{} \hat{U}^j_{n, m_n^{j,2,k},\hat{\ell}_n^{j,2,k}}
    \label{eq:estimate_utility_based_on_load}
    \end{equation}
    \item
    Choose a new Mood. If an agent was previously Content $S_n^{j,2,k-1}=C$, and its action and estimated utility have not changed $m_n^{j,2,k}=m_n^{j,2,k-1} \land \hat{U}_n^{j,2,k}=\hat{U}_n^{j,2,k-1}$, then it will remain Content with probability 1:
    \begin{equation}
    \label{eq:content_cotent}
        C \rarrow C
    \end{equation}
    If an agent was previously Discontent $S_n^{j,2,k-1}=D$, or changed its resource or estimated utility from the previous block $m_n^{j,2,k} \neq m_n^{j,2,k-1} \lor \hat{U}_n^{j,2,k} \neq \hat{U}_n^{j,2,k-1}$, its new Mood is chosen according to the following probability:
    \begin{equation}
    [C/D] \rightarrow \Biggl\{\begin{array}{cc}
    C & \rm{w.p.} \;\; \varepsilon^{U_{\max}-\hat{U}_n^{j,2,k}}\\
    D & \rm{w.p.} \;\; 1 - \varepsilon^{U_{\max}-\hat{U}_n^{j,2,k}}
    \end{array}. \label{eq:new_mood}
    \end{equation}
\end{enumerate}
{\bf Exploitation Phase: } The third and final phase of the ENE algorithm has only one block and $2^j$ time-steps. Each agent chooses individually and distributedly the resource it visited most frequency during the tail of the last Negotiation Phase:
\begin{align}
\label{eq:exp_resource}
    m_n^{j,3} =    \underset{m}{\arg \max} \;\; \sum _{k = (1-\alpha) j^{1+\delta/3}} ^{ j^{1+\delta/3}} \mathbb{1}(m_n^{j,2,k}=a)
\end{align}
where $0<\alpha<1$. The agent then stays with this resource throughout the block, gathering rewards. If the first two phases were successful, this phase will be Regret free. 
The complete ENE method is described in Algorithm \ref{alg:Estimate, Negotiation, Exploitation}.
\begin{algorithm}[!ht]
	\caption{The Estimation, Negotiation, and Exploitation algorithm at the individual agent level, to be performed fully distributedly and without communication between agents}
	\label{alg:Estimate, Negotiation, Exploitation}
	\begin{algorithmic}[1]
	\STATE Input: $ \varepsilon>0, \alpha \in (0,1), \delta>0, c\geq N$

	\FOR{$j=1$ \TO $J$ epochs}
	
    \STATE \textbf{Payoff Estimation Phase}
	    
    \FOR{$m = 1$ \TO $M$}
    \FOR{$\tau$ \TO $j$}
        \STATE $m_n^{j,1,m,\tau} \leftarrow m$
    \ENDFOR
    \STATE Estimate $U^j_{n,m,N}$ according to \eqref{eq:estimate_u_n_m_N}.
    \ENDFOR
    \FOR{$\tau$ \TO $j$}
        \STATE $m_n^{j,1,M+1,\tau} = \Biggl\{\begin{array}{cc}
        1 & \text{ w.p. } 1/2\\ 
        \emptyset &\text{ w.p. } 1/2
        \end{array}$
    \ENDFOR
    \STATE Calculate $\overline{r}_n^{j,1,M+1}$ according to \eqref{eq:estimate_u_n_M_plus_1}.
    \STATE Estimate $N$ according to \eqref{eq:estimation_of_the_number_of_agents}.
    \STATE Estimate $U^j_{n,m,\ell} \; \forall 1 \leq m \leq M, 
    1\leq \ell \leq N$ according to  \eqref{eq:estimate_u_n_m_ell}.

	\STATE \textbf{Negotiation Phase} 
	
	\STATE $S_n^{j,2,0} \leftarrow D$
	
	\FOR{$k$ \TO $j^{1+\delta/3}$}
	\STATE Choose new resource $m_n^{j,2,k}$ according to \eqref{eq:content_action_choice} or \eqref{eq:discontent_action_choice}.
	\FOR{$\tau$ \TO $j^{1+\delta/3}$}
	\STATE $m_n^{j,2,k,\tau} \leftarrow m_n^{j,2,k}$ 
	\ENDFOR
	\STATE Calculate $\overline{r}_n^{j,2,k} \leftarrow \frac{1}{j^{1+\delta/3}}\sum_{\tau=1}^{j^{1+\delta/3}} r_n^{j,2,k,\tau}$.
	\STATE Estimate load $\hat{\ell}_n^{j,2,k}$ according to \eqref{eq:estimate_load}.
	\STATE Estimate utility $\hat{U}_n^{j,2,k} $ according to \eqref{eq:estimate_utility_based_on_load}.
	\ENDFOR
	
	\STATE Choose new Mood according to \eqref{eq:content_cotent} or \eqref{eq:new_mood}.

    \STATE \textbf{Exploitation Phase} 
    \STATE Choose resource $m_n^{j,3}$ according to \eqref{eq:exp_resource}.
    \FOR{$\tau$ \TO $2^j$}
    \STATE $m_n^{j,3,1,\tau} \leftarrow m_n^{j,3}$
    \ENDFOR 
    \ENDFOR
	\end{algorithmic} 
\end{algorithm}


\section{Regret analysis of ENE algorithm}
\label{sec:regret}
In this section, we analyze the expected regret of the ENE  Algorithm  \ref{alg:Estimate, Negotiation, Exploitation}. We present the main Theorem, whose proof follows via  a sequence of Lemmas bounding the probability of error for each error event. 
\begin{theorem}
\label{thm:TotRegBnd}
For the resource sharing problem specified in Section \ref{sec:Cooperative Repeated Congestion games} there exists a parameter\footnote{The choice of $\gre$ is related to the perturbed Markov chain used in the Negotiation Phase. Practically, we found that values between $10^{-3}$ and $0.1$ perform satisfactorily.}  $\varepsilon>0$ in algorithm \ref{alg:Estimate, Negotiation, Exploitation} such that the regret of the ENE algorithm is upper-bounded by $O\left( \log_2^{3+\delta}(T)\right)$.
\end{theorem}



\begin{proof}
Let $R_1,R_2$ and $R_3$ denote the accumulated regret from the Estimation, Negotiation, and Exploitation phases of all the epochs of the algorithm, respectively, such that $R=R_1+R_2+R_3$.
Recall that Algorithm \ref{alg:Estimate, Negotiation, Exploitation} operates over $J$ epochs. By  Lemma \ref{lem:03} $R_3= O\left(J\right)$ and $R_1+R_2$ is upper bounded by
\begin{equation}
\label{eq:22}
\begin{split}
& NU_{\max}\sum_{j=1}^J \left( 
j(M+1)+j^{ 2+2\delta/3 } \right)
= O \left( J^{3+\delta}  \right)
\end{split}
\end{equation}
Furthermore, according to Lemma \ref{lem:num_epo_bound} $J\leq \log_2(T)$.
\end{proof}
Having sub-linear regret means that the ratio between the amount of time spent on sub-optimal allocations and the amount of time spent on the optimal allocation approaches zero as $T\rarrow\infty$. Because the sharing process described in Section \ref{sec:Cooperative Repeated Congestion games} is a variation on a single-agent multi-armed bandit problem, the optimal regret for this problem is $O(\log_2(T))$ according to \cite{lai1985asymptotically} and not very far from ours. 

\begin{lemma}
\label{lem:num_epo_bound} The number of epochs $J$ that Algorithm~\ref{alg:Estimate, Negotiation, Exploitation} operates satisfies
$E< \log_2(T)$
\end{lemma}
\begin{proof}
Ignoring the last epoch and the durations of the Estimation and Negotiation Phases produces $T\geq \sum_{j=1}^{J-1}  2^{j} = (2^J-2)$. 
\end{proof}
\begin{lemma}
$R_3 = O(J)$.
\label{lem:03}
\end{lemma}
\begin{proof}
The exploitation phase of epoch $j$ will accumulate regret  only if the following error event occurred:
\begin{equation}
\label{eq:err_1}
E^{j,3}: \mathbf{m}^{j,3} \neq \mathbf{m}^*
\end{equation}
That regret is upper bounded by $NU_{\max}2^j$. Therefore:
\begin{equation}
    R_{j,3} \leq NU_{\max}2^j \mathbb{P}\left(E^{j,3}\right)
\end{equation}
According to Lemma \ref{lem:exp_fail} the probability $\mathbb{P}\left(E^{j,3}\right)$ is $O(\exp(-j^{1+\delta/4}))$. Hence, $R_{j,3} = O(1)$. Finally, $R_3 = \sum_{j=1}^J R_{j,3} = O(J)$.
\end{proof}
\begin{lemma}
\label{lem:exp_fail}
$\mathbb{P}(E^{j,3})= O(\exp(-j^{1+\delta/4}))$.
\end{lemma}
\begin{proof}
Estimation and Negotiation Phase failures are denoted by
\begin{equation}
    \label{eq:est_fail}
    E^{j,1}: \mathbf{m}^{*j}\neq \mathbf{m}^{*}
\end{equation}
\begin{equation}
    \label{eq:neg_fail}
    E^{j,2}: \mathbf{m}^{j,3}\neq \mathbf{m}^{*j}
\end{equation}
The probabilities of these are bounded by $O(\exp(-j^{1.4}))$ and $O(\exp(-j^{1+\delta/4}))$, respectively, according to 
Lemmas \ref{lem:est_fail} and \ref{lem:neg_fail}, respectively.
Furthermore, 
\begin{equation}
    \mathbb{P}(E^{j,3})\leq \mathbb{P}(E^{j,1})+\mathbb{P}(E^{j,2})
\end{equation}
\end{proof}
\begin{lemma}
\label{lem:est_fail}
$\mathbb{P}(E^{j,1}) = O\left(\exp\left(-j^{1.4}\right)\right)$.
\end{lemma}
\begin{proof}
Let us define the following errors events:
\begin{align}
    & \tilde{E}^{j,1}: 
    \underset{n,m,\ell}{\max}  \left| U_{n,m,\ell}  - \hat{U}^j_{n,m,\ell} \right|  \geq \rho \label{}
    \\
    & \tilde{E}^{j,1}_N: \exists n : \hat{N}_n^j \neq N \label{}
    \\
    & \tilde{E}^{j,1}_U: \underset{n,m}{\max}  \left| U_{n,m,N}  - \hat{U}^j_{n,m,N} \right|  > \frac{\rho}{N} \label{eq:30}
\end{align}
We bound $\mathbb{P}\left(E^{j,1}\right)$ as follows:
\begin{equation}
\label{eq:sufficient_utility_estimation}
    \mathbb{P}(E^{j,1}) 
    \overset{(a)}{\leq}
    \mathbb{P}\left(\tilde{E}^{j,1}
    \right)
    \overset{(b)}{\leq}
    \mathbb{P}\left(\tilde{E}^{j,1}_N
    \right)
    +
    \mathbb{P}\left(\tilde{E}^{j,1}_U
    \right) 
\end{equation}
where (a) is a simple modification of Lemma (1) in \cite{bistritz2018distributed}, and (b) is clear from \eqref{eq:estimate_u_n_m_ell}. The probability $\mathbb{P}\left(\tilde{E}^{j,1}_N
\right)$ is $O(\exp(-j^{1.4}))$ according to Lemma \ref{lem:incorrect_estimation_of_N}. 
We bound $\mathbb{P}\left(\tilde{E}^{j,1}_U\right)$ as follows:
\begin{equation}
\label{eq:01}
\begin{split}
    & \mathbb{P}\left(\tilde{E}^{j,1}_U\right) 
    \overset{(a)}{\leq}
    NM
    \mathbb{P} \left( \left| U_{n,m,N}  - \hat{U}^j_{n,m,N} \right|  > \frac{\rho}{N}
    \right)
    \\ & 
    \overset{(b)}{\leq}
    2NM  \exp\left(- \frac{1}{2b}\left(\frac{\rho}{N}\right)^2
    \frac{j^2+j}{2} \right) 
    \\ & 
    \overset{(c)}{\leq}
    2NM  \exp\left(- \frac{j^2}{4b}\left(\frac{\rho}{N}\right)^2 
    \right) = O(\exp(-j^{1.5}))
\end{split}
\end{equation} 
Where (a) is a union bound on the agents and resources, (b) is Chernoff's inequality for the average of i.i.d. sub-Gaussian random variables, and (c) holds for any positive $j$.

\end{proof}
\begin{lemma}
\label{lem:incorrect_estimation_of_N}
$\mathbb{P}\left(\tilde{E}^{j,1}_N
    \right) = O(\exp(-j^{-1.4}))$.
\end{lemma}
\begin{proof}
An error event by agent $n$ in block $k$ of phase 1 of epoch $j$ is denoted with $E^{j,1,k}_n$ and defined by 
\begin{equation}
    \left|\overline{r}_n^{j,1,k} -
    \expect{\overline{r}_n^{j,1,k}}\right| 
    >
    2^{-N-4}
    \expect{\overline{r}_n^{j,1,k}} 
\end{equation}
We prove in the appendix that:
\begin{equation}
\label{eq:34}
\mathbb{P}\left(\hat{N}_n^j \neq N    \right) \leq \mathbb{P}\left(E^{j,1,1}_n\right)+\mathbb{P}\left(E^{j,1,M+1}_n\right)
\end{equation}

According to Chernoff's inequality:
\begin{equation}
\label{eq:04}
    \mathbb{P}\left(E^{j,1,1}_n\right) \leq 2\exp \left(-\frac{j^2}{4b} \left(U_{n,1,N} \cdot 2^{-N-4} \right)^2 \right)    
\end{equation}
According to Lemma \ref{lem:01} the probability $\mathbb{P}\left(E^{j,1,M+1}_n\right)$ is $O(\exp(-j^{1.4}))$.
A union bound on the agents preserves the asymptotic behavior in $j$ such that $\mathbb{P}\left(\tilde{E}^{j,1}_N
\right)$ is upper bounded by the same order.
\end{proof}
\begin{lemma}
\label{lem:01}
$\mathbb{P}(E_n^{j,1,M+1})= O(\exp(-j^{1.4}))$
\end{lemma}
\begin{proof}
Let the number of samples collected by agent $n$ during all epochs until epoch $j$ in phase 1 and block $M+1$ be:
\begin{equation}
     \xi^{j,1,M+1}_n \triangleq \sum_{i=1}^j \sum_{\tau=1}^i \mathbb{1}\left(m_n^{i,1,M+1,\tau} =1 \right)
\end{equation}
The probability $\mathbb{P}\left(E^{j,1,M+1}_n \right)$ is upper bounded by the sum of the following two probabilities:
\begin{align}
\mathbb{P}\left( \xi^{j,1,M+1}_n < \frac{j^{1.5}}{2} \right)& \label{eq:36}
\\
\mathbb{P}\left(E^{j,1,M+1}_n
\;\bigg|\; 
\xi^{j,1,M+1}_n > \frac{j^{1.5}}{2}
\right)&
\label{eq:37}
\end{align}
Since $\xi^{j,1,M+1}_n$ is a binomial random variable with parameters $\frac{1}{2}(j^2+j)$ and $1/2$, the probability in \eqref{eq:36} can be upper bounded with Hoeffding's inequality by the following expression:
\begin{equation}
\label{eq:102}
\begin{split}
    &\exp \left(
    -2\cdot \frac{1}{2}(j^2+j)
    \left( 
    \frac{1}{2} - \frac{j^{1.5}}{2\cdot \frac{1}{2}(j^2+j)}
    \right)^2
    \right)
    \\ \overset{(a)}{<}
    &\exp \left(
    -j^2
    \left( 
    \frac{1}{2} - \frac{1}{j^{0.5}}
    \right)^2
    \right)
    \overset{(b)}{<}
    \exp \left( \frac{-j^2}{100}
    \right)
\end{split}
\end{equation}
Where (a) holds for any epoch and (b) holds from the seventh epoch. This is of course $O(\exp(-j^{1.4}))$.

Since $\overline{r}_n^{j,1,M+1}$ is the sum of a sub-Gaussian random variable with variance proxy $b$ and another random variable bounded between $U_{n,1,1}$ and $U_{n,1,N}$ then 
the probability in \eqref{eq:37} can be upper bounded with Chernoff-Hoeffding's inequality by the following expression:
\begin{equation}
\begin{split}
    & 2\exp \left(-\frac{1}{4} \cdot \frac{\left(\expect{\overline{r}_n^{j,1,M+1}} 
    \cdot
    2^{-N-4}\right)^2 j^{1.5}}{b+(U_{n,1,1}-U_{n,1,N})^2}\right)
\end{split}
\label{eq:03}
\end{equation}
The last expression is also $O(\exp(-j^{1.4}))$.

\end{proof}

\begin{lemma}
\label{lem:neg_fail}
$\mathbb{P}(E^{j,2}) = O\left(\exp\left(-j^{1+\delta/4}\right)\right)$.
\end{lemma}

\begin{proof}
A Load Estimation Error occurs when at least one agent incorrectly estimated its load during at least one block of the Negotiation Phase:
\begin{equation}
    E^{j,2}_{\text{load}}:\exists n, k :
    \hat{\ell}^{j,2,k}_{n} \neq \ell^{j,2,k}_{n}
\label{eq:load_est_err}
\end{equation}
The probability of this error is $O\left(\exp\left(-j^{1+\delta/4}\right)\right)$ according to lemma \ref{lem:load_est_err_bnd}.

An Insufficient Mixing Time Error is defined as follows:
\begin{equation}
    E^{j,2}_{\text{mix}}= E^{j,2}
    \land    \neg E^{j,2}_{\text{load}}
\end{equation}
The probability of this error is $O\left(\exp\left(-j^{1+\delta/4}\right)\right)$ according to lemma \ref{lem:insuff_mix_time}. 

To finish this lemma:
\begin{equation}
\begin{split}
    \mathbb{P}\left(E^{j,2}\right) 
    \leq 
    \mathbb{P}\left(E^{j,2}_{\text{load}}
    \right)
    +
    \mathbb{P}\left(E^{j,2} _{\text{mix}}
    \right).
\end{split}
\end{equation}
\end{proof}
\begin{lemma}
\label{lem:load_est_err_bnd}
$\mathbb{P}\left(E^{j,2}_{\text{load}}
    \right)
    = 
    O\left(\exp\left(-j^{1+\delta/4}\right)\right)$
\end{lemma}
\begin{proof}
Agent $n$ will incorrectly estimate its load in block $k$ if the average reward it receives in that block with resource $m_n^{j,2,k}=m$ and load $\ell_n^{j,2,k}=\ell$ is sufficiently far from its estimated utility:
\begin{equation}
\label{eq: 35}
    \left| r_{n,m,\ell}^{j,2,k} - \hat{U}_{n,m,\ell}^{j} \right| > \Phi
\end{equation}
Where $\Phi$ is defined as 
\begin{equation}
    \Phi \triangleq \frac{1}{3} \underset{n,m}{\min} \left| \hat{U}^j_{n,m,N-1}  - \hat{U}^j_{n,m,N} \right|
\end{equation}

The difference between the average reward and estimated utility in the left hand side of \eqref{eq: 35} is a sub-Gaussian random variable with variance proxy $\frac{b}{j^{1+\delta/3}}+\frac{2b}{j^2+j}
=
b\frac{j+2j^{\delta/3}+1}{j^{2+\delta/3}+j^{\delta/3+1}}
<
\frac{2b}{j^{1+\delta/3}}$.
A union bound on  \eqref{eq: 35} with respect to the agents and load-estimation-blocks together with Chernoff's bound produces:
\begin{equation}
\begin{split}
    \mathbb{P}\left(E^{j,2}_{\text{load}}
    \right) \leq 2N  j^{1+\delta/3}   \exp \left( -  \frac{\Phi^2 j^{1+\delta/3}}{2b}  \right)
    \\
    = 
    O\left(\exp\left(-j^{1+\delta/4}\right)\right)
\end{split}
\end{equation}
\end{proof}
\begin{lemma}
$\mathbb{P}\left(E^{j,2}_{\text{mix}}\right)  = 
O\left(\exp\left(-j^{1+\delta/4}\right)\right)$
\end{lemma}

\begin{proof}
Let $\mathcal{M}^j$ be the Markov chain of the Negotiation Phase of epoch $j$. 
Define the optimal state to be the estimated optimal allocation \eqref{eq:est_opt_alloc} with the optimal estimated utilities \eqref{eq:opt_est_util} and all agents Content:
\begin{equation}
\label{eq: optimal state definition}
    z^* \triangleq (\mathbf{m}^{*j}, \mathbf{U}^{*j}, C^N).
\end{equation}

Let $\mathcal{\tilde{M}}^j$ be a Markov Chain on $\mathcal{Z}$ identical to $\mathcal{M}^j$ except that a Load Estimation Error is impossible. 
The stationary probability in $\mathcal{\tilde{M}}^j$ of the optimal state is $\pi^*$. According to lemma \ref{lem:sufficient_stationary_distribution}: $\pi^*>2/3$. 
The state of the Markov Chain in block $k$ of the Negotiation Phase of epoch $j$ is $z^{j,2,k}$.
The estimated stationary probability of  the optimal state is:
\begin{equation}
    \hat{\pi}^*  \triangleq \frac{1}{\alpha
    j^{1+\delta/3}}\sum_{k = (1-\alpha) j^{1+\delta/3}} ^{j^{1+\delta/3}} \mathbb{1}(z^{j,2,k} = z^*)
\end{equation}

We wish to bound the probability that the optimal state was visited in less than half the blocks of the tail of the Negotiation Phase due to insufficient mixing time of the Markov chain. According to lemma \ref{lem:insuff_mix_time} this is: 
\begin{equation}
\label{eq: 52}
\mathbb{P}\left( \hat{\pi}^* \leq \frac{1}{2} \leq \pi^*\frac{3}{4} \right) 
= O\left(\exp\left(-j^{1+\delta/4}\right)\right)
\end{equation}
\end{proof}
\begin{lemma}
\label{lem:sufficient_stationary_distribution}
$\pi^* >2/3$
\end{lemma}
\begin{proof}
Any two agents can increase or decrease each other's utilities by choosing the same action or a different action, respectively. 
Formally, this is known as interdependence and is defined in \cite{marden2014achieving}. 
As $\gre \rarrow 0$ we have $\pi^* \rarrow 1$. This is due to the interdependence of our dynamics and Theorem 3.2 from \cite{marden2014achieving}. For our $\varepsilon$, the stationary probability of the optimal state is $\pi^*>2/3$.
\end{proof}
\begin{lemma}
\label{lem:insuff_mix_time}
$\mathbb{P}\left( \hat{\pi}^* \leq \frac{1}{2} \leq \pi^*\frac{3}{4} \right) 
= O\left(e^{-j^{1+\delta/4}}\right)$
\end{lemma}
\begin{proof}
Let $\mathcal{Z}$ be the set of states of $\mathcal{\tilde{M}}^j$ and Let $\pi$ be its stationary distribution.
Let $\varphi$ be the distribution of block $(1- \alpha) j^{1+\delta/3}$. Let  $\left\Vert \varphi\right\Vert_\pi$ be the $\pi$-norm of $\varphi$ defined by:
\begin{equation}
    \left\Vert \varphi\right\Vert _{\pi}
    \triangleq
    \sqrt{\sum_{z\in \mathcal{Z}}\frac{\varphi^{2}(z)}{\pi (z)}}
\end{equation}

Let $T_{1/8}$ be the the minimal amount of time necessary for $\mathcal{\tilde{M}}^j$ to reach a total variation distance of $1/8$ from $\pi$ with arbitrary initialization. Let $C$ be a positive constant, independent of $\varphi$, $T_{1/8}$, and $\pi$. 
Let $\eta \triangleq 1-\frac{1}{2\pi^*}$. According to lemma \ref{lem:sufficient_stationary_distribution} we have  $\pi^*> 2/3$ and therefore $\eta<1$. After setting $\frac{1}{2} = (1-\eta)\pi^*$ into  \eqref{eq: 52}, the Markovian concentration bound from \cite{ChungMine} produces the following upper bound on the probability of an Insufficient Mixing Time Error of  \eqref{eq: 52}: 
\begin{equation}
\label{eq:final_equation_1}
    C||\varphi||_\pi \exp \left( - \frac{ \left(1 - \frac{1}{2 \pi^*} \right)^2 \pi^*\alpha j^{1+\delta/3}}{72T_{1/8}} \right)
\end{equation}
Because $\pi^*>2/3$ we obtain
\begin{equation}
    \left(1 - \frac{1}{2 \pi^*} \right)^2 \pi^* \geq \frac{1}{24}
\end{equation}
setting this into  \eqref{eq:final_equation_1} completes the proof. 

There is a tradeoff regarding $\left\Vert \varphi\right\Vert _{\pi}$. Starting from an arbitrary initial condition, $\left\Vert \varphi\right\Vert _{\pi}$
can be large.
By dedicating the first $\alpha j^{1+\delta/3}$ blocks of the Negotiation Phase to letting the Markov chain approach its stationary distribution,
and starting to count the visits to $z^*$ only afterward, we can reduce $\left\Vert \varphi\right\Vert _{\pi}$
significantly, at the cost
of $(1-\alpha) j^{1+\delta/3}$ turns less for estimating $z^*$. Optimizing over $\alpha$ can improve the constants of
the bound in \eqref{eq:final_equation_1}.
\end{proof}

\section{Numerical Examples}
\label{sec: simulation results}

We illustrate  the performance of the ENE algorithm in the non-asymptotic regime in comparison with a random allocation and the distributed (selfish) Upper Confidence Bound (dUCB) algorithm~\cite{besson2018multi}. 

Suppose $N=4$ agents share $M=2$  communication channels. 
When two agents or more choose the same channel at the same time, they share it equally via a round-robin Time Division Multiple Access (TDMA) mechanism. 
Assume one communication channel has high throughput (strongly desirable) and the other has low throughput (weakly desirable). If all agents use the high throughout channel, then the channel becomes over crowded and its throughout reduces.
We illustrate how the agents can achieve an optimal distributed allocation over the channels using the ENE algorithm.



To improve the convergence speed the number of blocks and their duration can be scaled by fixed constants to better fit the specific problem parameter, e.g., number of agents, number of resources, and noise variance. 
This is a feature of the ENE algorithm, and does not change the $O(\log_2^3(T))$ regret guarantee. 
In all the simulations the noise $\nu$ is Gaussian with variance $0.1$, the constant $\alpha$ is 1, the constant $c$ is $N$, and $\varepsilon = 10^{-2}$.


Figure \ref{fig:reg_eff}(a) depicts the sample path of the regret when the matrix $U_{n,m,1}$ is:
\begin{equation}
\label{eq:example_matrix}
    \begin{pmatrix}
    1 & 1 & 1 & 1 \\
    0.24 & 0.24 & 0.24 & 0.24
    \end{pmatrix}
\end{equation}
It can be seen that the regret is indeed $O(\log_2^3(t))$, while the two other algorithms suffer a linear regret. 
In this simulation, the Estimation and Negotiation Phase block sizes are scaled by $5\cdot 10^2$, the number of Negotiation Phase blocks is scaled by $5\cdot 10^3$, and the Exploitation block size is scaled by $2.5\cdot 10^6$. The number of epochs is $10$. 
 
Let the worst welfare be denoted by
\begin{equation}
    W_{\text{worst}} \triangleq  \displaystyle \min _{{\mathbf{m}} } W(\mathbf{m})
\end{equation}
We also measure the performance by the efficiency: 
\begin{equation}
    \text{Efficiency} \triangleq \frac{ \left(\frac{1}{T} \sum _{t=1} ^T W^t \right) - W_{\text{worst}}}{W^* - W_{\text{worst}}}
\end{equation}

Figure \ref{fig:reg_eff}(b) depicts the average efficiency of 50 random trials where the matrices $U_{n,m,1}$ are:
\begin{equation}
\label{eq:random_matrices}
\begin{split}
    \begin{pmatrix}
    1 & 1 & 1 & 1 \\
    0.2 & 0.2 & 0.2 & 0.2
    \end{pmatrix} - Q,
\end{split}
\end{equation}
where $Q$ is a random matrix where all the entries are uniformly distributed between 0 and 0.2. 
The steady-state efficiency of the ENE algorithm is 93\% while for the random allocation it is 85\% and for the selfish UCB it is only 82\%.
\begin{figure}[ht]
    \centering
    \includegraphics[width =  1 \columnwidth]{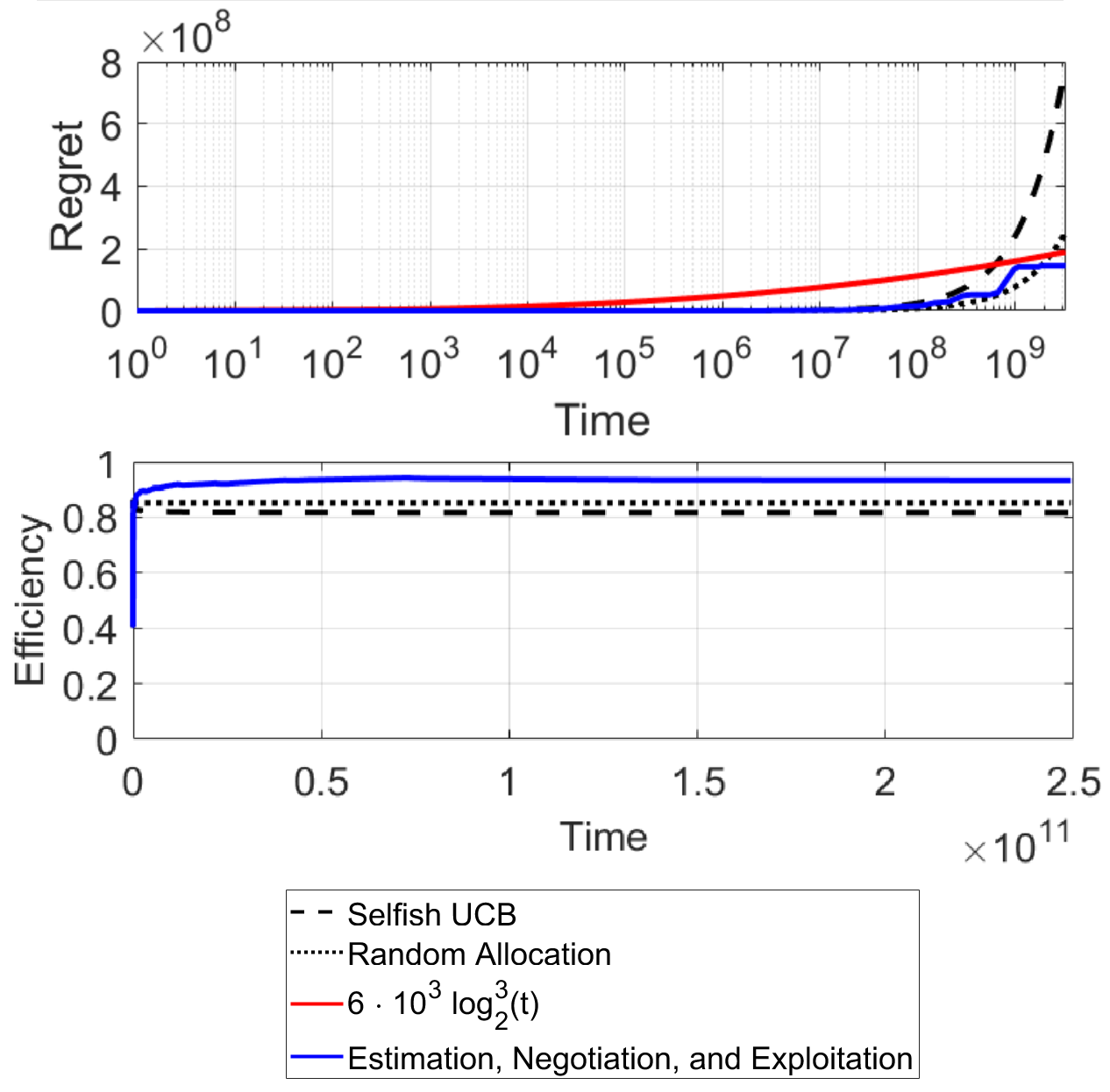}
    \caption{Comparing the ENE algorithm, the selfish UCB, and a random allocation over time. (a) Regret for the matrix \eqref{eq:example_matrix}. (b) Average Efficiency for 50 random matrices \eqref{eq:random_matrices}.}
    \label{fig:reg_eff}
\end{figure}

\section{Conclusions}
\label{sec: conclusions}
Our work extends the multi-agent multi-armed bandit problem formulation and the regret  analysis to the heavily congested case where multiple agents can share a resource without any communication between agents. 
The problem is motivated by applications to cloud computing, spectrum collaboration (in wireless networks), and machine scheduling (in operations research).  
Our formulation uses a congestion game with incomplete information. 
We propose a novel three-phased algorithm with provable expected regret of $O\left(\log_2^{3+\delta}(T)\right)$. 
We demonstrated how our approach leads to welfare maximization in this particular anonymous non-linear repeated game with agent-specific stochastic utilities. 


\bibliographystyle{IEEEbib}


\appendix
\section{Appendix}
\label{sec: supplementary}

\subsection{Properties of sub-Gaussian random variables}
Let $\nu$ be a sub-Gaussian random variable with variance proxy $b$. The moment-generating function of $\nu$ is bounded by
\begin{equation}
    \mathbb{E}(e^{\nu s}) \leq \exp\left(\frac{b s^2}{2}\right) \quad \forall s \in \mathbb{R} 
\end{equation}
The following properties are repeatedly used in the proofs of the lemmas of the main paper.
For proofs, see e.g., \cite{wainwright2019high}.
\begin{lemma}
\label{lem: sub-Gaussian times a constant}
Let $\nu$ be a sub-Gaussian random variable with variance proxy $b$. Let $g$ be some constant.
The random variable $g\nu $ is sub-Gaussian with variance proxy $bg^2$.
\end{lemma}
\begin{lemma}
\label{lem: sub-Gaussian average}
The average of $k$ independently and identically distributed sub-Gaussian random variables with variance proxy $b$ is sub-Gaussian with variance proxy $b/k$.
\end{lemma}


An immediate corollary of Lemmas \ref{lem: sub-Gaussian times a constant} and \ref{lem: sub-Gaussian average} is that the average of $k$ independently and identically distributed sub-Gaussian random variables with variance proxy $b$ times a constant $g$ is a sub-Gaussian random variable with variance proxy $\frac{b g^2}{k}$:
\begin{equation}
    \label{eq: sub-Gaussian average times a constant}
    \mathbb {E} \left( \exp\left(  g s \frac{1}{k}\sum_{i=1}^k \nu_i \right) \right) \leq \exp \left( \frac{b s^2 g^2 }{2k}\right)
\end{equation}

\begin{lemma}
\label{lem: sum of sub Gauss with diff parameters}
Let $\nu_1$ and $\nu_2$ be independent sub-Gaussian random variables with variance proxys $b_1$ and $b_2$ respectively, then $\nu_1 + \nu_2$ is sub-Gaussian with variance proxy $b_1+b_2$.
\end{lemma}

\begin{lemma}
\label{lem: sub-Gaussian Chernoff}
Let $\nu$ be sub-Gaussian random variable with variance proxy $b$. For any $s$, it holds that 
\begin{equation}
    \mathbb{P}(\nu > s) \leq e^{-\frac{s^2}{2b}} \qquad \text{and} \qquad \mathbb{P}(\nu <- s) \leq e^{-\frac{s^2}{2b}}
\end{equation}
\end{lemma}

An immediate corollary of Lemma \ref{lem: sub-Gaussian Chernoff} and equation \eqref{eq: sub-Gaussian average times a constant} is that the probability that a constant $g$ times the absolute value of the average of $k$ independently and identically distributed sub-Gaussian random variables with variance proxy $b$ will be greater than $s$ is bounded as follows:
\begin{equation}
    \label{eq: Chernoff on average of sub-Gaussians times a constant}
    \mathbb{P}\left( \left|   \frac{g}{k}\sum_{i=1}^k \nu_i  \right| > s \right) \leq  2\exp \left(  -\frac{k s^2 }{2 b g^2} \right).
\end{equation}
\subsection{Explaining \eqref{eq:estimate_u_n_M_plus_1}}
If agent $n$ was active in time slot $\tau$ of block $M+1$ of phase 1 of epoch $j$ its load is a random variable distributed as follows:
\begin{equation}
    \ell_n^{j,1,M+1,\tau} \sim \text{Binomial}(N-1,1/2) +1
\end{equation}
Let up consider the more general case where the probability of an agent to be active in time slot $\tau$ of block $M+1$ of phase 1 of epoch $j$ is not 1/2 but rather $p$.
The first inverse moment of this random variable can be obtained as follows:
\begin{equation}
\begin{split}
    \expect{\frac{1}{\ell_n^{j,1,M+1,\tau}}} 
     = 
    & \sum_{i=0}^{N-1}\frac{1}{i+1}\binom{N-1}{i}p^i(1-p)^{N-1-i}\\
     = 
    & \sum_{i=0}^{N-1}\frac{1}{N}\binom{N}{i+1}p^i(1-p)^{N-1-i}\\
     = 
    \frac{1}{pN} & \sum_{i=0}^{N-1} \binom{N}{i+1}p^{i+1}(1-p)^{N-1-i}\\
     = 
    \frac{1}{pN} & \sum_{i=0}^{N-1} \binom{N}{i+1}p^{i+1}(1-p)^{N-1-i}
    \\
    & +\frac{(1-p)^{N}}{pN}-\frac{(1-p)^{N}}{pN}\\
     = 
    \frac{1}{pN} & \sum_{i=0}^{N} \binom{N}{i}p^{i}(1-p)^{N-i}  -\frac{(1-p)^{N}}{pN} \\
    & = 
    \frac{1}{pN} -\frac{(1-p)^{N}}{pN}
     = \frac{1-(1-p)^{N}}{pN}
\end{split}
\end{equation}
Setting $p=1/2$ completes the proof. 
\subsection{Some details regarding the proof of lemma \ref{lem:incorrect_estimation_of_N}}

Agent $n$ will correctly estimate the number of agents when
\begin{equation}
    N-\frac{1}{2} 
    <
    \frac{1}{\ln{(1/2)}}\ln{\left(1 - \frac{\overline{r}_n^{j,1,M+1}}{2\overline{r}_n^{j,1,1}}\right)} 
    <
    N+\frac{1}{2}
\label{eq:61}
\end{equation}
After applying some algebra to \eqref{eq:61} we obtain:
\begin{equation}
\begin{split}
     2\left(1-\frac{1}{2^N}\right)
    +
    \frac{1}{2^{N-1}}
    \left(1-\sqrt{2}\right)
    \leq
    \frac{\overline{r}_n^{j,1,M+1}}
    {\overline{r}_n^{j,1,1}}
    \\ \leq
    2\left(1-\frac{1}{2^N}\right)
    +
    \frac{1}{2^{N-1}}
    \left(1-\frac{1}{\sqrt{2}}\right)      
\end{split}
\label{eq:62}
\end{equation}
Because $\left|1-2^{1/2}\right|>1-2^{-1/2}$ we can make the lower bound in \eqref{eq:62} tighter:
\begin{equation}
\begin{split}
    2\left(1-\frac{1}{2^N}\right)
    -
    \frac{1}{2^{N-1}}
    \left(1-\frac{1}{\sqrt{2}}\right)   
    \leq
    \frac{\overline{r}_n^{j,1,M+1}}
    {\overline{r}_n^{j,1,1}}
    \\ \leq
    2\left(1-\frac{1}{2^N}\right)
    +
    \frac{1}{2^{N-1}}
    \left(1-\frac{1}{\sqrt{2}}\right) 
\end{split}
\label{eq:63}
\end{equation}
We now wish to move from an additive bound to a multiplicative bound. 
We take notice of the following:
\begin{equation}
\label{eq:64}
    \frac{2^{1-N}(1-2^{-1/2})}{2(1-2^{-N})}
    >
    \frac{(1-2^{-1/2})}{2^N}
    >
    \frac{1/4}{2^N}
    >
    2^{-2-N},    
\end{equation}
We use \eqref{eq:64} to make the bounds in \eqref{eq:63} tighter:
\begin{equation}
\begin{split}
    2\left(1-\frac{1}{2^N}\right)
    \left(1-2^{-N-2}\right)
    \leq
    \frac{\overline{r}_n^{j,1,M+1}}
    {\overline{r}_n^{j,1,1}}
     \\ \leq
    2\left(1-\frac{1}{2^N}\right)
    \left(1+2^{-N-2}\right)    
\end{split}
\label{eq:65}
\end{equation}
We multiply all sides of \eqref{eq:65} by the ratio of the expectations:
\begin{equation}
    \left(1-2^{-N-2}\right)
    \leq
    \frac{\overline{r}_n^{j,1,M+1}}
    {\overline{r}_n^{j,1,1}} 
    \;
    \cdot
    \;
    \frac{\expect{\overline{r}_n^{j,1,1}}}
    {\expect{\overline{r}_n^{j,1,M+1}}}
    \leq
    \left(1+2^{-N-2}\right)
\label{}
\end{equation}
The last expression will hold if
\begin{align}
    \sqrt{
    \left(1-2^{-N-2}\right)}
    \leq
    & \frac{\overline{r}_n^{j,1,M+1}}
    {\expect{\overline{r}_n^{j,1,M+1}}}    \leq
    \sqrt{
    \left(1+2^{-N-2}\right)}
    \label{eq:66}
    \\
    \sqrt{
    \left(1-2^{-N-2}\right)}
    \leq
    & \frac
    {\expect{\overline{r}_n^{j,1,1}}}
    {\overline{r}_n^{j,1,1}} 
    \leq
    \sqrt{
    \left(1+2^{-N-2}\right)}
    \label{eq:67}
\end{align}

Subtracting 1 from all sides of \eqref{eq:66} and \eqref{eq:67} and re-arranging  produces:
\begin{equation}
\begin{split}
    & \expect{\overline{r}_n^{j,1,M+1}}
    \left( 
    {\sqrt{
    \left(1-2^{-N-2}\right)}}
    - 1 \right)
    \leq \\
    \;& \qquad\quad \overline{r}_n^{j,1,M+1}
    -
    \expect{\overline{r}_n^{j,1,M+1}}
    \leq \\
    &\expect{\overline{r}_n^{j,1,M+1}}
    \left( 
    {\sqrt{
    \left(1+2^{-N-2}\right)}}
    - 1 \right)
    \label{eq:69}    
\end{split}
\end{equation}

\begin{equation}
\begin{split}
    & \expect{\overline{r}_n^{j,1,1}}
    \left( 
    \frac{1}{\sqrt{
    \left(1-2^{-N-2}\right)}}
    - 1 \right)
    \geq \\
     \; & \qquad\quad
    \overline{r}_n^{j,1,1}
    -
    \expect{\overline{r}_n^{j,1,1}}
    \geq \\
    & \expect{\overline{r}_n^{j,1,1}}
    \left( 
    \frac{1}{\sqrt{
    \left(1+2^{-N-2}\right)}}
    - 1 \right)
    \label{eq:70}    
\end{split}
\end{equation}

For convenience we define the following bounds and notice that the upper bound is tighter in \eqref{eq:69} while the lower bound is tighter is \eqref{eq:70}:
\begin{equation}
\label{eq:71}
\begin{split}
    L_1 \triangleq \left|
    \frac{1}{\sqrt{
    \left(1+2^{-N-2}\right)}}
    - 1
    \right|
    \leq \\
    \frac{1}{\sqrt{
    \left(1-2^{-N-2}\right)}}
    - 1
\end{split}
\end{equation}
\begin{equation}
\label{eq:72}
\begin{split}
    L_2 \triangleq
    {\sqrt{
    \left(1+2^{-N-2}\right)}}
    - 1
    \leq \\
    \left|
    {\sqrt{
    \left(1-2^{-N-2}\right)}}
    - 1
    \right|    
\end{split}
\end{equation}

We make the bounds in \eqref{eq:69} and \eqref{eq:70} tighter according to our observations from \eqref{eq:71} and \eqref{eq:72}:
\begin{equation}
    \left|
    \overline{r}_n^{j,1,1}
    -
    \expect{\overline{r}_n^{j,1,1}}
    \right|
    \leq
    \expect{\overline{r}_n^{j,1,1}}L_1
\end{equation}
\begin{equation}
    \left|
    \overline{r}_n^{j,1,M+11}
    -
    \expect{\overline{r}_n^{j,1,M+1}}
    \right|
    \leq
    \expect{\overline{r}_n^{j,1,M+1}}L_2
\end{equation}
We use Taylor's series to obtain the following bounds:
\begin{alignat}{1}
    \frac{1}{\sqrt{1+x}} &\leq 1 - 0.5x +0.375x^2
    \label{eq:75}
    \\
    \sqrt{1+x} &\geq 1+x/2-x^2/8
    \label{eq:76}
\end{alignat}
We use \eqref{eq:75} and \eqref{eq:76} to obtain the following bounds:
\begin{alignat}{2}
    L_1 & 
    \geq 1-1+2^{-N-3}-0.375 \cdot 2^{-2N-4} 
    &&\geq
    2^{-N-4}
    \label{eq:77}
    \\
    L_2 & 
    \geq
    1+2^{-N-3} - 2^{-2N-7}-1
    &&\geq
    2^{-N-4}
    \label{eq:78}
\end{alignat}
When \eqref{eq:77} and \eqref{eq:78} hold, then \eqref{eq:61} holds. 

\subsection{Final comment on the Negotiation Phase duration and further explanation of equation  \eqref{eq:22}}
We wish to explain why in \eqref{eq:22} we have $\sum_{j=1}^J j^{2+2\delta/3} = O(J^{3+\delta})$. 
To do so we begin with the following equation:
\begin{equation}
\label{eq:79}
    (j^{1+\delta/3}-1)^3 = j^{3+\delta} -3j^{2+2\delta/3}+3j^{1+\delta/3}-1
\end{equation}
Rearrange \eqref{eq:79} to obtain:
\begin{equation}
\begin{split}
    3j^{2+2\delta/3}-3j^{1+\delta/3}+1
    =
    j^{3+\delta}
    -
    (j^{1+\delta/3}-1)^3
    \leq \\
    j^{3+\delta}
    -
    (j-1)^{3+\delta}
\end{split}
\label{eq:80}
\end{equation}
We take \eqref{eq:80} and sum over all the epochs of the algorithm :
\begin{equation}
\label{eq:81}
    \sum_{j=1}^J
    3j^{2+2\delta/3}-3j^{1+\delta/3}+1
    \leq
    \sum_{j=1}^J
    j^{3+\delta}
    -
    (j-1)^{3+\delta}
    =
    J^{3+\delta}
\end{equation}
We rearrange \eqref{eq:81} to obtain
\begin{equation}
    \sum_{j=1}^J
    j^{2+2\delta/3}
    \leq
    \frac{J^{3+\delta}-J}{3}+
    \sum_{j=1}^J
    j^{1+\delta/3}=O(J^{3+\delta})
    \label{}
\end{equation}

\end{document}